\documentclass[11pt,reqno]{amsart}
\usepackage{amssymb,amsfonts,amsmath,bbm,amssymb,amsthm}
\usepackage{hyperref}
\usepackage{graphicx}
\usepackage{subfigure,color}
\usepackage{enumerate}

\newtheorem{theorem}{Theorem}[section]
\newtheorem*{theorem*}{Theorem}
\newtheorem{definition}[theorem]{Definition}
\newtheorem{proposition}[theorem]{Proposition}
\newtheorem{corollary}[theorem]{Corollary}
\newtheorem{lemma}[theorem]{Lemma}

\def\Sc{{\mathcal S}}
\def\bra{\langle}
\def\ket{\rangle}

\def\be{\begin{eqnarray}}
\def\ee{\end{eqnarray}}

\DeclareMathOperator*{\trace}{Tr}
\DeclareMathOperator{\im}{Im}
\DeclareMathOperator{\id}{id}

\begin{document}

\title{Quantum channels with polytopic images and image additivity}
\author{Motohisa Fukuda}
\address{MF, IN, MMW: Zentrum Mathematik, M5, Technische Universit\"at M\"unchen, Boltzmannstrasse 3, 85748 Garching, Germany}
\author{Ion Nechita}
\address{IN: CNRS, Laboratoire de Physique Th\'{e}orique, IRSAMC, Universit\'{e} de Toulouse, UPS, F-31062 Toulouse, France}
\email{m.fukuda@tum.de, nechita@irsamc.ups-tlse.fr, m.wolf@tum.de}

\author{Michael M. Wolf}

\begin{abstract}
We study quantum channels with respect to their image, i.e., the image of the set of density operators under the action of the channel. We first  characterize the set of quantum channels having polytopic images and show that additivity of the minimal output entropy can be violated in this class.
We then provide a complete characterization of quantum channels $T$ that are \emph{universally image additive} in the sense that for any quantum channel $S$, the image of $T \otimes S$ is the convex hull of the tensor product of the images of $T$ and $S$. These channels turn out to form a strict subset of entanglement breaking channels with polytopic images and a strict superset of classical-quantum channels. 
\end{abstract}

\maketitle

\tableofcontents

\section{Introduction}
\label{sec:intro}
This work is motivated by the study of additivity problems in quantum information theory. These problems typically arise whenever tensor products appear in optimization problems. The fact that quantum states on a tensor product space are not restricted to convex combinations of product states enables various manifestations of the whole being more (or occasionally less) than the sum of its parts. A paradigm of such an additivity problem concerns the minimum output entropy of a quantum channel. Looking at the problem, which is described in detail in Section \ref{sec:additivity}, one may realize that the details of the channel are only relevant if they influence the image of the channel and how it behaves under tensorization. However, images of quantum channels and their behaviour under taking tensor products are poorly understood. The present paper attempts to makes a first step in the direction of improving this.  

Our main results can be stated informally as follows (for a precise statement, see Thm.\ref{thm:polytopic-image} and Thm.\ref{thm:uia}).

\begin{theorem*}[Channels with polytopic images]
A quantum channel $T$ has polytopic image if and only if it can be decomposed as the sum of a classical-quantum channel $T_1$ and an arbitrary channel $T_2$, where $T_{1,2}$ act on orthogonal diagonal blocks of the input, and the image of $T_2$ is included in the image of $T_1$.
\end{theorem*}

\begin{theorem*}[Universally image additive channels]
A quantum channel $T$ is universally image additive if and only if it is \emph{essentially classical quantum}, i.e., if it is entanglement breaking with POVM operators having unit norm.
\end{theorem*}

The paper is organized as follows. Section \ref{sec:def-examples} fixes the notation, defines the basic objects and briefly discusses some of their elementary properties. In Section \ref{sec:polytopic-image} we then derive a characterization of quantum channels (or more general linear maps) with polytopic images. The core of the section is a careful analysis of how a vertex can appear in the image of a channel. Section \ref{sec:additivity} then shows that, despite their seemingly 'classical' image, quantum channels with polytopic images can still violate the additivity of the minimum output entropy. In Section \ref{sec:uia-channels} a stronger form of additivity, namely universal image additivity, is analyzed and the set of channels with this property is characterized. Section \ref{sec:fixed-points} provides a preparatory result for this analysis. Finally, in Section \ref{sec:polytopic-EB} we reinvestigate entanglement breaking channels with polytopic images.

\section{Definition and first examples}
\label{sec:def-examples}

We denote by $M_d(\mathbb C)$ the set of $d\times d$ matrices with complex entries and define $\Sc_d\subseteq M_d(\mathbb C)$ to be the subset of density matrices, i.e., positive semi-definite matrices with unit trace. Occasionally, \emph{pure quantum states} described by rank one density matrices will be characterized by a corresponding unit vector in $\mathbb{C}^d$.

A \emph{positive operator valued measure} (POVM) will in this paper be identified with a tuple $(M_1,\ldots,M_k)$ of positive semi-definite operators $M_i\in M_d(\mathbb{C})$ for which $\sum_{i=1}^k M_i=\mathbbm{1}$.
 
A linear map $T:M_d(\mathbb C) \to M_n(\mathbb C)$ will be called a \emph{quantum channel} (in the Schr\"odinger picture) if it is completely positive and trace preserving.
Its \emph{image} will be defined as 
$$\mathrm{Im}(T): = T(\Sc_d) =  \{T(\rho)\, :\, \rho \in\Sc_d\} \subset M_n(\mathbb C).$$
Note that $\Sc_d$ and therefore also $\mathrm{Im}(T)$ are compact convex sets and that every extreme point of $\mathrm{Im}(T)$ has a pure state in its preimage.

Every linear map $T:M_d(\mathbb C) \to M_n(\mathbb C)$ can be represented as 
\begin{equation}\label{eq:genform}
T(\rho) = \sum_{i=1}^k \mathrm{Tr}(M_i \rho) \sigma_i,
\end{equation} with suitable matrices $M_i\in M_d(\mathbb{C})$ and $\sigma_i\in M_n(\mathbb{C})$. Depending on the properties of these matrices we get different classes of quantum channels, which we summarize in the following definition:

\begin{definition}[Classes of entanglement breaking channels]
Let $T:M_d(\mathbb C) \to M_n(\mathbb C)$ be a linear map. 
\begin{enumerate}
\item $T$ is an \emph{entanglement breaking quantum channel} if it admits a representation as in Eq.(\ref{eq:genform}) where the $M_i$'s form a POVM and the $\sigma_i$'s are density matrices.
\item $T$ is called an \emph{essentially classical-quantum channel} if it admits a representation as in Eq.(\ref{eq:genform}) where the $M_i$'s form a POVM for which $\forall i: ||M_i||=1$ and the $\sigma_i$'s are density matrices.\footnote{Here $||\cdot||$ denotes the operator norm.}
\item $T$ is called a \emph{classical-quantum (CQ) channel} if there is an orthonormal basis $\{e_i\in\mathbb{C}^d\}$ such that $T$ admits a representation as in Eq.(\ref{eq:genform}) with $k=d$ and where $M_i=e_ie_i^*$ and $\sigma_i\in\Sc_n$ for all $i$.
\end{enumerate}
\end{definition}

The image of an essentially classical-quantum channel can be seen to be the convex hull of the points $\{\sigma_1, \ldots, \sigma_k\}$ 
$$\mathrm{Im}(T) = \operatorname{hull}(\sigma_1, \ldots, \sigma_k).$$ That is, in this case $\mathrm{Im}(T)$ is a convex polytope.
In order to see this, note that the POVM operators associated to such a channel can be decomposed as 
\begin{equation}\label{eq:Mi-tildeMi}
M_i = e_i e_i^* + \tilde M_i,
\end{equation}
where $\{e_1, \ldots, e_k\}$ is an orthonormal family of vectors in $\mathbb C^d$ spanning a $k$-dimensional subspace $V$ and $(\tilde M_1, \ldots \tilde M_k)$ form a POVM on the orthogonal subspace $V^\perp$. In particular, 
\begin{equation}\label{eq:orth-tildeMi-ej}
\forall i,j, \quad \mathrm{Tr}(\tilde M_i e_j e_j^*) = 0.
\end{equation}
By choosing $\rho = e_ie_i^*$ we now get $\sigma_i \in \mathrm{Im}(T)$ while the the inclusion $\mathrm{Im}(T) \subseteq \operatorname{hull}(\sigma_1, \ldots, \sigma_k)$ follows from the POVM condition. 

Note that for general entanglement breaking quantum channels, the image may not be a convex polytope. As an example, consider the following case where the POVM operators are up to rescaling projections on non-orthogonal vectors: let $T:M_2(\mathbb C) \to M_3(\mathbb C)$ be of the form
\begin{equation}\label{eq:non-commuting-POVM} 
T(\rho) = \sum_{j=1}^3 \operatorname{Tr}(M_j \rho) e_je_j^*,
\end{equation}
where $\{e_j\}_{j=1}^3$ is an orthonormal basis of $\mathbb C^3$. The POVM operators $M_j \in M_2(\mathbb C)$ are defined by 
$$M_j = \frac{2}{3} P_{\omega^j} ,$$
where $P_{\omega^j}$ is the orthogonal projection on the complex number $\omega^j \in \mathbb C$ (seen as a vector in $\mathbb R^2 \subset \mathbb C^2$) and $\omega$ is a third root of unity. An easy computation shows that the image of $T$ is the set of diagonal matrices with entries inside the filled circle in Figure \ref{fig:non-commuting-POVM}. 

\begin{figure}[htbp] 
\includegraphics{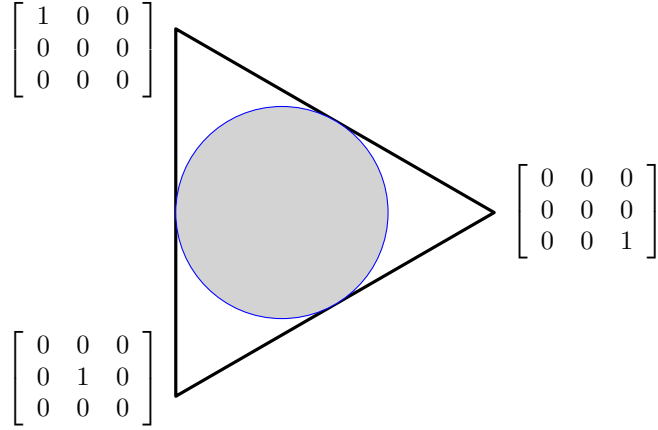}
\caption{The image of the entanglement breaking quantum channel $T$ from \eqref{eq:non-commuting-POVM} is not a polytope.} 
\label{fig:non-commuting-POVM}
\end{figure}

\section{Quantum channels with polytopic images}
\label{sec:polytopic-image}

In this section, we study the structure of quantum channels sending the convex body of quantum states to a convex polytope. Since the main applications of the results below are in quantum information theory, we state our results for quantum channels; note however that all the results in the current section remain valid in the case of linear maps between matrix algebras. 

We start by recalling the definition of a vertex of a convex set (see \cite[Definition 11.6.1]{ber}, and Figure \ref{fig:vertex} for a graphical example).

\begin{definition} Let $V$ be a finite dimensional real vector space and $C\subseteq V$ a closed convex set.
A point $x$ on the boundary of $C$ is called a \emph{vertex} if the intersection of all supporting hyperplanes of $C$ at $x$ is the set $\{x\}$. In particular, $x$ is then an extreme point of $C$.
\end{definition}

For convex polytopes, the sets of vertices and extreme points coincide.  A convex set like $\Sc_d$, on the other hand, has no vertices. If the image of a quantum channel has a vertex, then (since it is in particular an extreme point) there is always a pure input state that is mapped onto it.

\begin{figure}[htbp] 
\includegraphics{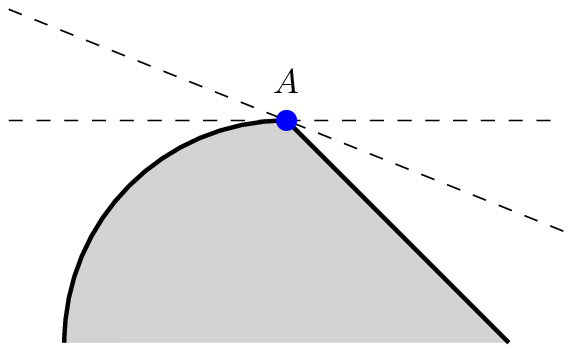} \quad\quad\quad\quad\quad
\includegraphics{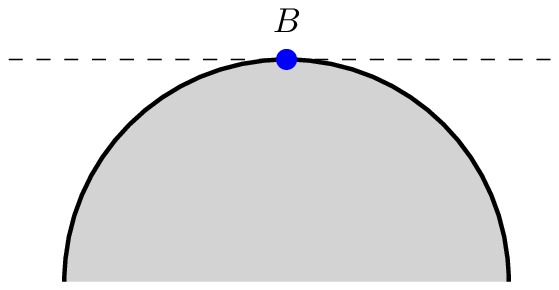}
\caption{On the left, the point $A$ is a vertex of the convex set, since the the number of independent supporting hyperplanes going through $A$ matches the dimension of the ambient vector space. On the right, there is a unique supporting hyperplane going through the point $B$, so $B$ is not a vertex.} 
\label{fig:vertex}
\end{figure}

\begin{lemma}\label{lem:elli}
Consider a quantum channel $T:M_d(\mathbb C) \to M_n(\mathbb C)$ and two different pure states $x,y \in \mathbb C^d$ such that $T(xx^*)$ is a vertex of $\mathrm{Im}(T)$. If $H$ is the two-dimensional subspace of $\mathbb C^d$ spanned by $x$ and $y$, then the restriction $T_H$ of $T$ to $\mathrm{End}(H)$ is of the form 
$$T_H(\rho) = \langle x, \rho x \rangle T(xx^*) + \langle {x}_\perp, \rho x_\perp \rangle T(x_\perp x_\perp^*),$$
where $x_\perp$ is orthogonal to $x$ in $H$. Moreover, if $T(yy^*)$ is also a vertex of $\mathrm{Im}(T)$, different from $T(xx^*)$, then $x \perp y$.
\end{lemma}
\begin{proof}
Let $\Sc_H:=\Sc_d\cap\mathrm{End}(H)$ be the set of quantum states on $H$, which is isomorphic to the Bloch ball (see Appendix \ref{sec:Bloch-ball}). Using the linearity of the map $T$, it follows that $T(\Sc_H)$, as the image of a sphere under a linear map, is an ellipsoid $E$.

We shall now show that the ellipsoid $E$ is degenerate, being at most $1$-dimensional. Let $C$ be the intersection of the image of $T$ with the smallest affine space containing $E$. Obviously, we have $T(xx^*) \in E \subseteq C$. Moreover, $T(xx^*)$ is a vertex of $C$. Since $E$ is a subset of $C$ that contains $T(xx^*)$, the point $T(xx^*)$ must be a vertex of $E$ as well and thus $E$ must be either a point or a segment.

When the dimension of the ellipsoid is zero, the statement of the theorem is trivially true, so we focus on the case where the ellipsoid is a segment with $T(xx^*)$  being one of the end points. If $\sigma \neq T(xx^*)$ denotes the other endpoint of the segment, then the restriction $T_H$ can be written as
\be\label{eq:segment}
T_H(\rho) = f(\rho) T(xx^*) + (1- f(\rho)) \sigma
\ee
for any quantum state $\rho$ supported on $H$.  
Moreover, taking the inner product with some matrix orthogonal to $\sigma$ 
we notice that the function $f(\cdot)$ is linear. 
By the Riesz representation theorem, there exists a Hermitian matrix $M\in \mathrm{End}(H)$ with $0 \leq M \leq I$ and such that 
\be
f(\rho) = \trace [\rho M]
\ee
Using $\rho = xx^*$, we must have $xx^* \leq M$. Also, $f(\rho) = 0$ for some other input $\rho$ (in order to get the segment image), but then $\rho$ must  be $x_\perp x_\perp^*$, so that $M = xx^*$, and the proof of the first part of the statement is now complete. 

In the case where $T(yy^*)$ is also a vertex, the image must be the line segment $[T(xx^*), T(yy^*)]$ and thus $\sigma = T(yy^*)$ and $y=x_\perp$. 
\end{proof}

\begin{corollary}\label{cor:vecpace}
Let $T:M_d(\mathbb{C})\rightarrow M_n(\mathbb{C})$ be a quantum channel, $\sigma$ a vertex of $\mathrm{Im}(T)$ and $X_k:=\{x_1,\ldots,x_k\}$ a set of unit vectors in $\mathbb{C}^d$ such that $T(x_ix_i^*)=\sigma$ for all $i$. Then $T(xx^*)=\sigma$ for all unit vectors $x\in\mathrm{span}(X_k)$.
\end{corollary}
\begin{proof}
We can assume that the $x_i$'s are linearly independent. The proof is done by induction over $k$. For $k=1$ the statement is evident. Suppose it is true for all $x\in\mathrm{span}(X_k)$ and $T(x_{k+1}x_{k+1}^*)=\sigma$. Then let $\tilde{x}\in X_{k+1}$ be a unit vector orthogonal to $X_k$ and decompose  $x_{k+1}=c_1x+c_2\tilde{x}$  with $c_1,c_2\in\mathbb{C}$ and $x\in\mathrm{span}(X_k)$. Lemma \ref{lem:elli} applied to $H=\mathrm{span}(x,\tilde{x})$ and $\rho=x_{k+1}x_{k+1}^*$ leads to $T(\tilde{x}\tilde{x}^*)=\sigma$. Consequently, for any unit vector $y\in X_{k+1}$ we obtain $T(yy^*)=\sigma$ when applying Lemma \ref{lem:elli} within $H=\mathrm{span}(y,\tilde{x})$.
\end{proof}

\begin{proposition}\label{prop:channel-vertex-decomposition}
Let $T:M_d(\mathbb C) \to M_n(\mathbb C)$ be a quantum channel whose image has vertices $\{\sigma_i\}_{i=1}^k$. Then, the input Hilbert space $\mathbb C^d$ decomposes as
$$\mathbb C^d = \left[ \bigoplus_{i=1}^k V_i \right] \oplus W = V \oplus W,$$
in such a way that
\begin{equation}\label{eq:channel-vertex-decomposition}
\forall \rho \in \Sc_d: \qquad T(\rho) = T_1(\rho_V) + T_2(\rho_W) = \left[ \sum_{i=1}^k \mathrm{Tr}(P_{V_i}\rho) \sigma_i \right] + T_2(\rho_W).
\end{equation}
In the above expression, for a given subspace $X \subseteq \mathbb C^d$, $P_X$ denotes the orthogonal projection onto $X$ and $\rho_X:=P_X\rho P_X$ is the restriction of $\rho$ to $X$. Note that the channel $T_1$ in \eqref{eq:channel-vertex-decomposition} is classical-quantum.
\end{proposition} 
\begin{proof}
Start by defining, for all $i=1,\ldots , k$
\begin{equation}\label{eq:def-Vi}
V_i  = \mathrm{span} \{x \in \mathbb C^d \, : ||x||=1\ \wedge \ T(xx^*) = \sigma_i\}.
\end{equation}
Using the second part of Lemma \ref{lem:elli}, we deduce that the spaces $V_i$ are orthogonal, and we define then $W=V_0$ to be the orthogonal complement of their direct sum. 

Let us choose an orthonormal basis $\{x_i\}_{i=1}^d$ of $\mathbb C^d$ that respects the direct sum decomposition above. Consider now two basis elements $x_i,x_j$, such that $x_i \in V_r$, $x_j \in V_s$, with $r,s=0, 1, \ldots, k$, $r \neq s$. 
By applying Lemma \ref{lem:elli} and Corollary \ref{cor:vecpace} to the subspace $H_{i,j} = \mathrm{span} \{x_i, x_j \}$, the restriction of $T$ to $H_{i,j}$ can be written as
\be
T_{H_{i,j}} (\rho)  = \bra x_i,  \rho  x_i \ket T(x_ix_i^*) + \bra x_j , \rho  x_j \ket  T(x_jx_j^*).
\ee
This implies that the action of the map $T$ is block-diagonal with respect to the direct sum decomposition of $\mathbb C^d$.
 In particular, $T(\rho)=\sum_{i=0}^k T(P_{V_i}\rho P_{V_i})$ and Eq.(\ref{eq:channel-vertex-decomposition}) follows by linearity.
\end{proof} 

This leads to a complete characterization of maps with polytopic images:

\begin{theorem}[Channels with polytopic images]\label{thm:polytopic-image}
Let $T:M_d(\mathbb C) \to M_n(\mathbb C)$ be a quantum channel whose image is a convex polytope with $k$ vertices $\{\sigma_i\}_{i=1}^k$. 
Then, there exists a subspace $V \subseteq \mathbb C^d$ such that 
\be
T(\rho) = T_1 (\rho_V) + T_2 (\rho_{V^\perp}) 
\ee
where $T_1$ is a classical-quantum channel written as in \eqref{eq:channel-vertex-decomposition} and $T_2$ is such that $\mathrm{Im} (T_2) \subseteq \mathrm{Im} (T_1)$ and, for all $i$, $\sigma_i \notin \mathrm{Im}(T_2)$. Conversely, every map of this form has polytopic image with vertices $\{\sigma_i\}$.
\end{theorem} 
\begin{proof}
The channel decomposition and the fact that $T_1$ is a classical-quantum channel  follow from Proposition \ref{prop:channel-vertex-decomposition}. The inclusion of images follows easily from the fact that 
$$\mathrm{hull}(\{\sigma_i\}_{i=1}^k)=\mathrm{Im}(T_1) = \mathrm{Im}(T) = \bigcup_{\lambda\in [0,1]} (1-\lambda) \mathrm{Im}(T_1) + \lambda\mathrm{Im}(T_2).$$
The fact that the vertices $\sigma_i$ do not belong to the image of $T_2$ follows again from Proposition \ref{prop:channel-vertex-decomposition} since the only pure quantum states $x$ such that $T(xx^*) = \sigma_i$ belong to $V_i \subseteq V$, see Eq.\eqref{eq:def-Vi}.
\end{proof}

\begin{corollary}[Dimension of polytopic images] Let $T:M_d(\mathbb{C})\rightarrow M_n(\mathbb{C})$ be a quantum channel whose image is a convex polytope with $k$ vertices. Then
$$\mathrm{dim}\big(\mathrm{Im}(T)\big)\leq k-1\leq d-1.$$
\end{corollary}
\begin{proof}
While the first inequality holds for all convex polytopes, the second is a consequence of the direct sum decomposition in Prop.\ref{prop:channel-vertex-decomposition}.
\end{proof}
Note, that if we do not constrain the dimension of the input space, then any convex polytope $C\subseteq\Sc_n$ can be obtained as the image of a classical-quantum channel: if $\{\sigma_i\in\Sc_n\}_{i=1}^k$ is the set of vertices of $C$, we can choose $d=k$ and an orthonormal basis $\{e_i\}$ in $\mathbb{C}^d$ such that $T(\rho)=\sum_i\langle e_i,\rho e_i\rangle \sigma_i$ will have $\mathrm{Im}(T)=C$.

\section{Additivity of minimum output entropy}
\label{sec:additivity}

In the previous section, quantum channels having polytopic image have been shown to be the sum of a CQ channel and an arbitrary channel whose image is included in the image of the CQ channel. Since the minimum output entropy of a channel is a function of its image, one may think that the additivity properties of channels with polytopic images should be identical to the ones of CQ channels. In this section, we show that this intuition is false, by constructing examples of quantum channels that have polytopic images but violate the additivity of the minimum output ($p$-R\'enyi) entropy. The intuition behind the counterexamples is the fact that the image of the CQ channel can ``hide'' the image of a non-additive channel, and taking tensor products might reveal the non-additivity. 

Recall that the $p$-R\'enyi entropy of a probability distribution $x=(x_1, \ldots, x_n)$ is defined, for all $p \geq 1$, by
$$H^{(p)}(x) = \frac{\log \sum_{i=1}^n x_i^p}{1-p},$$
where the value at $p=1$ is obtained by taking the limit $p \to 1$ and it is equal to the Shannon entropy
$$H(x) = H^{(1)}(x) = \sum_{i=1}^n x_i \log x_i.$$
These definitions extend, via functional calculus, to density matrices $\rho \in \Sc_n$, the corresponding quantities being called the \emph{R\'enyi entropies} of $\rho$; the value at $p=1$ is called the \emph{von Neumann} entropy of the density matrix $\rho$. For a quantum channel $T:M_d(\mathbb C) \to M_n(\mathbb C)$, the \emph{minimum output $p$-R\'enyi entropy} is defined by
$$H^{(p)}_{\min} (T) = \min_{\rho \in \Sc_d} H^{(p)} (T(\rho)).$$
Note that the minimum above is attained for rank one projectors, since the entropy functions are concave. The additivity of the minimum output entropy $H^{(p)}_{\min}$ plays an important role in quantum information theory  and was, for a long time, an open question: given two quantum channels $T_{1,2}$, is it true in general that
\begin{equation}\label{eq:additivity}
H^{(p)}_{\min} (T_1 \otimes T_2) = H^{(p)}_{\min} (T_1) + H^{(p)}_{\min} (T_2) \quad ?
\end{equation}

Celebrated results by Hastings \cite{has09}, Hayden and Winter \cite{hwi} show that, for all $p \geq 1$, there exist quantum channels that do not satisfy the additivity condition above. However, CQ channels always satisfy the additivity equation \eqref{eq:additivity}; more generally, all entanglement breaking channels are additive \cite{sho02}. It is thus natural to ask whether quantum channels having polytopic images are always additive. A negative answer is provided in the following, by ``hiding'' Hastings-Hayden-Winter counterexamples behind CQ channels. 
Since entanglement breaking channels have additive minimal output entropies, the resulting channels cannot be entanglement breaking despite the fact that they have polytopic image. 
 
\begin{theorem}
For all $p \geq 1$, there exist dimensions $d,n$, and a quantum channel $T:M_d(\mathbb C) \to M_n(\mathbb C)$ having a polytopic image, such that the pair $(T , \bar T)$ violates the additivity relation \eqref{eq:additivity}. 
\end{theorem}
\begin{proof}
The starting point of the proof is the idea to hide the image of a given additivity-violating channel inside a polytope, which in turn can be realized as an image of some CQ channel.   
In order for this to be possible, it is useful if the given image is in the interior of the state space, i.e., if all output states have full rank.
 
To restrict ourselves to such examples, we use a result from \cite{cn2}, which states that random quantum channels typically do not have any zero eigenvalues for large enough $d \in \mathbb N$ if $\lim_{d \to \infty}  \frac{d}{N(d)} <1 $, where $N(d)$ is the dimension of the environment. 
On the other hand, in such an asymptotic regime, the violation of additivity was proven to be typical in \cite{fk,asw,fuk,cn2} (see also \cite{has09, bho, aswR}).
That is, there exists some channel $T_2:M_{d_2}(\mathbb C) \to M_n(\mathbb C)$ that violates the additivity equation \eqref{eq:additivity} for some $\varepsilon >0$:
$$
H^{(p)}_{\min} (T_2 \otimes \bar T_2) < 2 H^{(p)}_{\min} (T_2) - 2\varepsilon.$$
Here, $\bar T_2$ is the complex conjugate of $T_2$, which is defined by taking complex conjugate  Kraus operators. 
Importantly, we can find such a $T_2$ where all the output states have full rank by taking the intersection of the two typical phenomena.

Now we hide the image of $T_2$ inside some sufficiently tight polytopic image. By standard arguments there is an arbitrary good outer approximation of $\mathrm{Im}(T_1)$ by a convex polytope $P\subseteq\Sc_n$ (cf.\cite{Gruber}). In particular, we can construct a CQ channel 
 $T_1:M_{d_1}(\mathbb C) \to M_n(\mathbb C)$ that has $P$ as its image and by
choosing $P$ sufficiently close to $\mathrm{Im}(T_2)$ (at the cost of increasing $d_1$) we can guarantee that 
$$H^{(p)}_{\min} (T_1) \geq  H^{(p)}_{\min} (T_2) - \varepsilon.$$
As before, define $T:M_{d_1+d_2}(\mathbb C) \to M_n(\mathbb C)$ by 
$$T(\rho) = T_1(\rho_1) + T_2(\rho_2),$$
where $\rho_1$ and $\rho_2$ are the two diagonal blocks of $\rho$ having respective dimensions $d_1$ and $d_2$. We note that $H^{(p)}_{\min} (T)  = H^{(p)}_{\min} (T_1) $ and we have
$$H^{(p)}_{\min} (T \otimes \bar T) \leq H^{(p)}_{\min} (T_2 \otimes \bar T_2) < 2 H^{(p)}_{\min} (T_2) - 2\epsilon 
\leq 2 H^{(p)}_{\min} (T_1) = 2 H^{(p)}_{\min} (T),$$
which shows that the pair $(T, \bar T)$ violates the additivity relation \eqref{eq:additivity}.
\end{proof}

\section{Fixed points of entanglement breaking channels}
\label{sec:fixed-points}

In order to prepare ourselves for the study of image additive channels in the subsequent section, we have to make a brief excursion and investigate the fixed point structure of entanglement breaking channels.

Throughout this section, we fix an entanglement breaking channel $T:M_d(\mathbb C) \to M_d(\mathbb C)$ having the same input and output space. 

\begin{lemma}\label{lem:no-algebra}
An entanglement breaking channel cannot leave invariant a matrix algebra of dimension $r \geq 2$.
\end{lemma}
\begin{proof}
Let $V = \mathrm{span}\{e_1, \ldots, e_r\} \subset \mathbb C^d$ be a $r \geq 2$ dimensional subspace such that the entanglement breaking channel $T:M_d(\mathbb C) \to M_d(\mathbb C)$ leaves invariant all operators in $\mathrm{End}(V)$. Consider the entangled state $\psi \in \mathbb C^d \otimes \mathbb C^d$
\begin{equation}
\psi = \frac{1}{\sqrt r} \sum_{i=1}^r e_i \otimes e_i.
\end{equation}
Then 
\begin{equation}
[T \otimes \mathrm{id}](\psi \psi^*) = \frac{1}{ r} \sum_{i,j=1}^r T (e_i e_j^*) \otimes e_i e_j^* = \frac{1}{ r}\sum_{i,j=1}^r e_i e_j \otimes e_i e_j^* =  \psi \psi^*,
\end{equation}
which is a contradiction, since $\psi$ is entangled and $T$ was supposed to break entanglement.
\end{proof}

Following \cite[Section 3]{lin} (see also \cite{wol}), define 
\begin{equation}
T_\infty = \lim_{N \to \infty} \frac 1 N \sum_{n=1}^N T^{n}.
\end{equation}
Since our linear maps live in a compact, finite-dimensional space, the above limit exists and defines a quantum channel that satisfies
\begin{equation}
T_\infty = T \circ T_\infty = T_\infty \circ T =  T_\infty^{2}.\end{equation}
Thus, $T_\infty$ is a completely positive projection on its image, the set of fixed points of $T$
\begin{equation}
\mathcal F_T = \{ X \in M_d(\mathbb C) \, : \, T(X)=X\}.
\end{equation}

The support subspace $V_T$ of $T_\infty(\mathbbm{1}_d)$ plays an important role, since the restriction of $T$ to this subspace has a full-rank invariant state. Using the structure theorem for finite dimensional $\ ^*$-algebras, one can obtain the following result \cite{lin,wol,bnpv}. 

\begin{proposition}\label{prop:fixed-point-general}
Given a quantum channel $T:M_d(\mathbb C) \to M_d(\mathbb C)$, there exist quantum states $\sigma_1, \ldots , \sigma_k \in \Sc_d$ having orthogonal supports such that
\begin{equation}\label{eq:fixed-point-general}
\mathcal F_T = 0_{V_T^\perp} \oplus \bigoplus_{i=1}^k M_{d_i}(\mathbb C) \otimes \sigma_i. 
\end{equation}
\end{proposition}

\begin{theorem}\label{thm:fixed-point-eb}
Let $T:M_d(\mathbb C) \to M_d(\mathbb C)$ be an \emph{entanglement breaking} quantum channel. The set of fixed points of $T$ is spanned by density matrices $\sigma_1, \ldots, \sigma_k$ with orthogonal supports
\begin{equation}
\mathcal F_T = \mathrm{span}\{\sigma_1, \ldots, \sigma_k\}
\end{equation}
and the channel $T_\infty$ that projects on $\mathcal F_T$ is \emph{essentially classical-quantum}
\begin{equation}\label{eq:Tinfty-ecq} 
T_\infty(\rho) = \sum_{i=1}^k \mathrm{Tr}(M_i\rho) \sigma_i,
\end{equation}
so that the $M_i$ form a POVM with $\|M_1\| = \cdots = \|M_k\| = 1$.
\end{theorem}
\begin{proof}
Note that since $T$ is entanglement breaking, $T_\infty$ is also entanglement breaking. The first statement follows now from the general characterization of fixed-point sets \eqref{eq:fixed-point-general} and Lemma \ref{lem:no-algebra}, which imply that the dimensions $d_i$ appearing in Proposition \ref{prop:fixed-point-general} have to be trivial, $d_i=1$.

Since the image of the non-restricted channel $T_\infty$ is $\mathcal F_T$, there exists a POVM $(M_1, \ldots, M_k)$ such that equation \eqref{eq:Tinfty-ecq} in the statement holds. The norm equalities are obtained using the fact that the $\sigma_i$ are fixed points:
\begin{equation}
1 \geq \|M_i\| \geq \mathrm{Tr}(M_i \sigma_i) = 1.
\end{equation}

\end{proof}

\section{Universally image additive channels}
\label{sec:uia-channels}

In this section, we introduce the notion of \emph{image additive} channels, which is stronger than the usual notion of additivity of minimum output $p$-R\'enyi entropies, introduced in Section \ref{sec:additivity}, see \eqref{eq:additivity}. We then prove the second main result of this paper, Theorem \ref{thm:uia}, which provides a characterization of quantum channels $T$ that are image additive with any other quantum channel $S$. Note that the same task for the usual notion of additivity seems rather difficult for the following reason. On one hand, entanglement breaking channels are additive (in the usual sense) with any other channel; on the other hand, the same holds true for the identity channel (or any other unitary conjugation, for that matter). These channels are of very different nature, and thus the set of quantum channels that are additive with all other channels is not likely to admit an easy description. 

We start with the definition of a pair of image additive channels and of universally image-additive channels.

\begin{definition}[Image additivity]
Two quantum channels $T_i:M_{d_i}(\mathbb C) \to M_{n_i}(\mathbb C)$ , $i=1,2$ are called \emph{image additive} if one of the following equivalent statements is satisfied:
\begin{enumerate}[i.]
\item The image of $T_1 \otimes T_2$ is the convex hull of the tensor product of the images of $T_{1,2}$:
$$\mathrm{Im}(T_1 \otimes T_2) = \mathrm{hull}\left[  \mathrm{Im}(T_1) \otimes \mathrm{Im}(T_2) \right];$$
\item For every unit vector $\psi \in \mathbb C^{d_1} \otimes \mathbb C^{d_2}$, there is a \emph{separable} state $\rho_{sep} \in \Sc_{d_1d_2}$ such that
$$[T_1 \otimes T_2](\psi \psi^* ) = [T_1 \otimes T_2](\rho_{sep}).$$
\end{enumerate}

A channel $T:M_d(\mathbb C) \to M_n(\mathbb C)$ is called \emph{universally image additive} if for all $n',d'\in\mathbb{N}$ and all channels $S:M_{d'}(\mathbb C) \to M_{n'}(\mathbb C)$, the pair $(T,S)$ is image additive.
\end{definition}

Note that image additivity is stronger than minimum $p$-R\'enyi entropy additivity. We now state the main result of this section, a characterization of universally image additive quantum channels.

\begin{theorem}\label{thm:uia}
Let $T:M_d(\mathbb C) \to M_n(\mathbb C)$ be a quantum channel. The following assertions are equivalent:
\begin{enumerate}[a)]
\item \label{it:usa} $T$ is universally image additive;
\item \label{it:sai} $T$ and $\mathrm{id}:M_d(\mathbb C) \to M_d(\mathbb C)$ are image additive;
\item \label{it:S-eb} There exists an entanglement breaking channel $S:M_d(\mathbb C) \to M_d(\mathbb C)$ such that $T=T \circ S$;
\item \label{it:S-ecq} There exists an essentially classical-quantum channel $S:M_d(\mathbb C) \to M_d(\mathbb C)$ such that $T=T \circ S$;
\item \label{it:cq} $T$ is essentially classical-quantum.
\end{enumerate}
\end{theorem}
\begin{proof}
We shall prove the series of implications \ref{it:usa} $\implies$ \ref{it:sai} $\implies$ \ref{it:S-eb} $\implies$ \ref{it:S-ecq} $\implies$ \ref{it:cq} $\implies$ \ref{it:usa}.

The implication \ref{it:usa} $\implies$ \ref{it:sai} is trivial. 
For \ref{it:sai} $\implies$ \ref{it:S-eb}, choose $\psi$ to be the maximally entangled state in $\mathbb C^d \otimes \mathbb C^d$. Then, using the hypothesis, there exists a separable state $\rho_{sep}$ such that 
$$[T \otimes \mathrm{id}](\psi \psi^*) = [T \otimes \mathrm{id}](\rho_{sep}).$$
Moreover, $\rho_{sep}$ is such that $\mathrm{Tr}_{1} \rho_{sep} = \mathbbm{1}_d$. Hence, thanks to the Choi-Jamio{\l}kowski isomorphism \cite{cho, jam}, there exists an entanglement breaking channel $S$ such that 
$$\rho_{sep} = [S \otimes \mathrm{id}](\psi \psi^*).$$
We have thus 
$$[T \otimes \mathrm{id}](\psi \psi^*) = [T \circ S \otimes \mathrm{id}](\psi \psi^*),$$
so that, using again the Choi-Jamio{\l}kowski isomorphism, $T = T \circ S$, with $S$ entanglement breaking. 

Let us now show \ref{it:S-eb} $\implies$ \ref{it:S-ecq}. Starting from $T = T \circ S$, we get, by recurrence, $T = T \circ S^{n}$ for all $n\geq 1$ and thus $T = T \circ S_\infty$. But $S$ is entanglement breaking, so, using Theorem \ref{thm:fixed-point-eb}, we have that $S_\infty$ is essentially classical-quantum.

For \ref{it:S-ecq} $\implies$ \ref{it:cq}, given an essentially CQ channel $S(\rho) =  \sum_{i=1}^k  \mathrm{Tr}(M_i \rho) \sigma_i$, note that
$$T(\rho) = [T \circ S](\rho) =  \sum_{i=1}^k  \mathrm{Tr} \left[ M_i\rho \right] T(\sigma_i),$$
which shows that $T$ is essentially CQ itself.

Finally, for  \ref{it:cq} $\implies$ \ref{it:usa}, consider an essentially CQ channel where we separate the unit-norm part from the operators $M_i$ and write
$$T(\rho) = \sum_{i=1}^k  \mathrm{Tr} \left[ (e_ie_i^* + \tilde M_i )\rho \right] \sigma_i.$$
Using \eqref{eq:orth-tildeMi-ej}, we get $T = T \circ S$, where $S$ is the essentially CQ (and thus entanglement breaking) channel
$$S(\rho) = \sum_{i=1}^k   \mathrm{Tr} \left[ (e_ie_i^* + \tilde M_i )\rho \right] e_ie_i^*.$$
For any channel $T_2$, we write
$$[T \otimes T_2](\psi \psi^*) = [T \otimes T_2] \left( [S \otimes \mathrm{id}](\psi \psi^*) \right) =  [T \otimes T_2] (\rho_{sep}),$$
with the separable input
$$\rho_{sep} = [S \otimes \mathrm{id}](\psi \psi^*).$$
\end{proof}

\section{On entanglement breaking channels with polytopic image}
\label{sec:polytopic-EB}

The main results of sections \ref{sec:polytopic-image} and \ref{sec:uia-channels} are, respectively, a characterization of quantum channels with polytopic image and identifying the set of image additive channels, as a subclass of entanglement-breaking channels. It turns out that essentially classical-quantum channels, which exhibit a strong form of additivity, have polytopic image. It is thus natural to ask whether these channels are precisely the entanglement-breaking channels having polytopic image. In this section, we construct two examples of entanglement breaking quantum channels having polytopic image, which are not essentially classical quantum. 
Before we introduce our examples, we need some preparatory lemmas. 
\begin{lemma}\label{lem:eCQ-ext-hull}
Consider an essentially classical-quantum channel $T:M_d(\mathbb C) \to M_n(\mathbb C)$ given by
$$T(\rho) = \sum_{i=1}^k \operatorname{Tr}(M_i \rho) \sigma_i,$$
with $\|M_i\|=1$, for all $i$. Since the image of $T$ is the convex hull of the quantum states $\sigma_i$, one can find a subset of these states, say the first $r$ of them, that are the extreme points of the image.
Then, there exist POVM operators $N_1, \ldots, N_r$, with $\|N_i\|=1$ for all $i$, such that
$$T(\rho) = \sum_{i=1}^r \operatorname{Tr}(N_i \rho) \sigma_i.$$
\end{lemma}
\begin{proof}
The result follows easily by decomposing each quantum state $\sigma_j$ that is not an extreme point of the image as a convex combination of those that are extreme points, and collecting the corresponding effect operators $M_j$. The unit norm property is a consequence of the fact that, for all $i \leq r$, $N_i \geq M_i$.
\end{proof}

\begin{lemma}\label{lem:EB-direct-sum}
Let $T:M_{d_1+d_2}(\mathbb C) \to M_n(\mathbb C)$ be a quantum channel defined by $T(\rho) = T_1(\rho_1) + T_2(\rho_2)$, where $T_i:M_{d_i}(\mathbb C) \to M_n(\mathbb C)$ are quantum channels and $\rho_i \in M_{d_i}(\mathbb C)$  is the upper-left (resp. lower-right) diagonal block of $\rho$:
$$\rho=\begin{bmatrix}
\rho_1 & \tau \\
\tau^* & \rho_2
\end{bmatrix}.$$
Then, the channel $T$ is entanglement breaking if and only if both channels $T_{1,2}$ are entanglement breaking.
\end{lemma}
\begin{proof}
By construction, the image of $T \otimes \id$ has the form: 
\be
\im (T\otimes \id) = \bigcup_{0\leq\lambda\leq1} \left[ \lambda \im(T_1 \otimes \id) + (1-\lambda) \im(T_2 \otimes \id) \right]
\ee
Then, obviously $\im(T \otimes \id)$ is separable if and only if $\im(T_1 \otimes \id)$ and $\im(T_2 \otimes \id)$ are both separable. 
\end{proof}

To introduce our first example, consider unital qubit channels that, in the Bloch ball picture, have a disc of radius $r$ as image. 
Then, the conditions for complete positivity \eqref{eq:FA-1} and \eqref{eq:FA-2} amount to $r \leq 1/2$, so the only discs centered at the origin which can be images of unital qubit quantum channels are the ones with radius smaller than one half. 

\begin{proposition}\label{prop:disc-image}
Let $T_2 : M_2(\mathbb C) \to M_2(\mathbb C)$ be a unital quantum channel that, in the Bloch ball representation, has a disc $D$ of radius $1/2$ as its image. Consider, for some integer $d$, a classical quantum channel $T_1:M_d(\mathbb C) \to M_2(\mathbb C)$ given by
$$T_1(\rho_1) = \sum_{i=1}^d \langle e_i, \rho_1 e_i \rangle \, \sigma_i,$$
where $\{e_i\}_{i=1}^d$ is an orthonormal basis of $\mathbb C^d$ and $\{\sigma_i\}_{i=1}^d$ is a set of \emph{mixed} quantum states forming the vertices of a convex polytope such that
$$D \subseteq \operatorname{hull}\{\sigma_i\}_{i=1}^d.$$
Define a quantum channel $T:M_{d+2}(\mathbb C) \to M_2(\mathbb C)$ by $T(\rho) = T_1(\rho_1) + T_2(\rho_2)$, where $\rho_{1,2}$ are the upper (resp. lower) diagonal blocks of $\rho$. Then, the channel $T$ has a polytopic image and it is entanglement breaking, but not essentially classical-quantum. 
\end{proposition}
\begin{proof}
The fact that the image of $T$ is a polytope comes from the inclusion of the images of the channels $T_{1,2}$. Moreover, it follows from \cite[Theorem 4]{rus} that the channel $T_2$ is entanglement breaking  and thus $T$ is also entanglement breaking, by Lemma \ref{lem:EB-direct-sum}. 

We show now, by reductio ad absurdum, that the channel $T$ is not essentially classical-quantum. Suppose $T$ were an essentially classical-quantum channel; since the image of $T$ is a convex polytope having as extreme points precisely the (mixed) quantum states $\{\sigma_i\}_{i=1}^d$, it follows by Lemma \ref{lem:eCQ-ext-hull} that there exists a POVM $(M_i)_{i=1}^d$ consisting of operators of norm one, such that
$$T(\rho) = \sum_{i=1}^d \operatorname{Tr}(M_i \rho) \sigma_i.$$
Restricting to the lower-right $2 \times 2$ corner of $M_{d+2}(\mathbb C)$, we can write
$$T_2(\rho_2) = \sum_{i=1}^d \operatorname{Tr}(N_i \rho_2) \sigma_i,$$
for any qubit density matrix $\rho_2$.
Here $N_i$ is the lower-right $2\times 2$ block of $M_i$. 
Note that the $N_i$'s again form a qubit POVM but the norms may be smaller than one.

Since all the states $\sigma_i$ are mixed (and thus have full rank), one can find a positive number $\varepsilon$ such that
$$\forall 1 \leq i \leq d, \qquad \sigma_i':= (1 + \varepsilon) \sigma_i - \varepsilon \frac{\mathbbm{1}} 2 \geq 0.$$
Writing the channel $T_2$ in terms of the new matrices $\sigma_i'$, we obtain
$$T_2(\rho_2) = \frac{1}{1+\varepsilon}  \sum_{i=1}^d \operatorname{Tr}(N_i \rho_2) \sigma_i' + \frac{\varepsilon}{1+\varepsilon} \frac{{\mathbbm{1}}}{2} = \frac{1}{1+\varepsilon} T_2'(\rho_2) + \frac{\varepsilon}{1+\varepsilon} \Delta(\rho_2),$$
where $\Delta$ denotes the totally depolarizing channel, and 
$$T_2'(\rho_2) := \sum_{i=1}^d \operatorname{Tr}(N_i \rho_2) \sigma_i'$$
is again a quantum channel. Since the image of the channel $T_2$ is a disc of radius one half, the image of $T_2'$ is a disc of radius $(1+\varepsilon)/2$, which contradicts the Fujiwara-Algoet conditions \eqref{eq:FA-1}-\eqref{eq:FA-2}, completing the proof.
\end{proof}

Let us also consider a three dimensional version of the previous example, by looking  at  qubit depolarizing channels
$$\Delta_r(\rho) = r \rho + (1-r) \frac {\mathbbm{1}} 2.$$
Such a channel has a sphere of radius $r$ as an image and it is thus entanglement breaking if and only if $r \leq 1/3$ (see \cite[Theorem 4]{rus}).

\begin{proposition}\label{prop:sphere-image}
Let $T_2 = \Delta_{r} : M_2(\mathbb C) \to M_2(\mathbb C)$ be a qubit depolarizing channel with parameter $r=1/3$; the image of $T_2$ is a sphere $S$ of radius $1/3$. Consider, for some integer $d\in\mathbb{N}$, a classical quantum channel $T_1:M_d(\mathbb C) \to M_2(\mathbb C)$ given by
$$T_1(\rho_1) = \sum_{i=1}^d \langle e_i, \rho_1 e_i \rangle \, \sigma_i,$$
where $\{e_i\}_{i=1}^d$ is an orthonormal basis of $\mathbb C^d$ and $\{\sigma_i\}_{i=1}^d$ is a set of \emph{mixed} quantum states that are the vertices of a convex polytope such that
$$S \subseteq \operatorname{hull}\{\sigma_i\}_{i=1}^d.$$
Define a quantum channel $T:M_{d+2}(\mathbb C) \to M_2(\mathbb C)$ by $T(\rho) = T_1(\rho_1) + T_2(\rho_2)$, where $\rho_{1,2}$ are the upper (resp. lower) diagonal blocks of $\rho$. Then, the channel $T$ has a polytopic image and it is entanglement breaking, but not essentially classical-quantum. 
\end{proposition}
\begin{proof}
As in the previous proposition, the polytopic image and the entanglement breaking properties follow directly from the construction. Let us now show, in the same manner as before, that the channel $T$ is not essentially classical-quantum. Suppose, for a contradiction, that $T$ were an essentially classical-quantum channel. Using the same reasoning as before, one can find density matrices $\{\sigma_i'\}$ and a POVM $(N_i)$ such that
$$T_2(\rho_2) = \frac{1}{1+\varepsilon}  \sum_{i=1}^d \operatorname{Tr}(N_i \rho_2) \sigma_i' + \frac{\varepsilon}{1+\varepsilon} \frac{{\mathbbm{1}}}{2} = \frac{1}{1+\varepsilon} T_2'(\rho_2) + \frac{\varepsilon}{1+\varepsilon} \Delta(\rho_2),$$
where $\Delta$ denotes the totally depolarizing channel, and 
$$T_2'(\rho_2) := \sum_{i=1}^d \operatorname{Tr}(N_i \rho_2) \sigma_i'$$
is an entanglement breaking quantum channel. Recall now that $T_2$ is the depolarizing channel of parameter $r=1/3$ and thus
$$T_2'(\rho_2) =  \frac{1+\varepsilon}{3} \rho_2 + \frac{2-\varepsilon}{3} \frac {\mathbbm{1}} 2 = \Delta_{(1+\varepsilon)/3}(\rho_2),$$
which is a contradiction, since the qubit depolarizing channel of parameter $r=(1+\varepsilon)/3$ is not entanglement breaking. 
\end{proof}

Let us finally summarize the picture for quantum channels having polytopic images of small dimensions: 

\begin{proposition}
Consider quantum channels having a polytope image of dimension $r \geq 0$. Then, if $r=0,1$, the channel has to be classical-quantum and thus entanglement-breaking. If $r=2,3$, there exist examples of entanglement breaking channels with polytopic image that are not essentially classical-quantum. 
\end{proposition}
\begin{proof}
If $r=0$, then $\operatorname{Im}(T) = \{\sigma\}$ and the channel $T$ is simply
$$T(\rho) = \operatorname{Tr}(\rho) \sigma,$$
which is classical-quantum. 

If $r=1$, use Theorem \ref{thm:polytopic-image} to decompose $T$ as
$$T(\rho) = T_1(\rho_1) + T_2(\rho_2),$$
where $T_1$ is a classical-quantum channel and $T_2$ is any channel such that
$$\operatorname{Im}(T_2) \subseteq \operatorname{Im}(T_1) = \operatorname{Im}(T).$$
Hence the image of $T_2$ is also a segment (possibly degenerated) and one can use repeatedly Theorem \ref{thm:polytopic-image} to decompose this channel into classical-quantum channels. This procedure ends after a finite number of steps, since the dimension of the input space is finite. 

The cases $r=2,3$ are treated in Propositions \ref{prop:disc-image}, \ref{prop:sphere-image}.
\end{proof}

\bigskip

\noindent \textit{Acknowledgements.}
We would like to thank Toby Cubitt for useful discussions. The authors are grateful for financial support via the CHIST-ERA/BMBF project CQC and from the John Templeton Foundation (ID$\sharp$48322). I.N.'s research has been supported by a von Humboldt fellowship and by the ANR projects {OSQPI} {2011 BS01 008 01} and {RMTQIT}  {ANR-12-IS01-0001-01}. M.M.W. acknowledges the Isaac Newton Institute for Mathematical Sciences (Cambridge, UK) where this work has been finalized. 

\appendix 
\section{Bloch ball representation}
\label{sec:Bloch-ball}

We briefly recall the Bloch ball. 
Any qubit state $\rho$ can be written as
\be
\rho = \frac 12 {\mathbbm{1}} + \sum_{i=1}^3 w_i \sigma_i
\ee
where $\{\sigma_i\}_{i=1}^3$ are the Pauli matrices and $\{w_i\}_{i=1}^3 \subset \mathbb R$ are such that $\sum_{i=1}^3 w_i^2 \leq 1$. 
So, the set of quantum qubit states can be identified as the unit ball in $\mathbb R^3$.
Hence, up to rotations, which correspond to unitaries on the Hilbert space, any trace preserving linear map on qubit sates can expressed in terms of two vectors $\lambda,t\in\mathbb{R}^3$ corresponding to compressions and shifts, respectively, w.r.t. the corresponding axis.
Hence, the image is always an ellipsoid such that for $i \in \{1,2,3\}$ the radius along the axis $i$ is $\lambda_i$ 
and the center is $(t_1,t_2,t_3) \in \mathbb R^3$. 
When the channel is unital, $t_i=0$ holds for all $i$ and 
the Fujiwara-Algoet conditions \cite{fal} show that the inequalities:
\begin{align}
\label{eq:FA-1}\lambda_1 + \lambda_2 &\leq 1+\lambda_3 \\
\label{eq:FA-2}\lambda_1 - \lambda_2 &\leq 1-\lambda_3
\end{align}
are equivalent to  complete positivity of the map.


\begin{thebibliography}{9}

\bibitem{aswR}
Aubrun, G., Szarek, S. and Werner, E. 
{\it Non-additivity of R\'enyi entropy and Dvoretzky's theorem.}
J. Math. Phys. 51, 022102 (2010). 

\bibitem{asw}
Aubrun, G., Szarek, S. and Werner, E. 
{\it Hastings's additivity counterexample via Dvoretzky's theorem.}
Comm. in Math. Phys. 305, no. 1, 85--97 (2011).

\bibitem{ber}
Berger, M.
{\it Geometry I.}
Springer-Verlag (1987).

\bibitem{bnpv}
Blume-Kohout, R., Ng, H.K., Poulin, D., and Viola, L.
{\it Information preserving structures: A general framework for quantum zero-error information.}
Phys. Rev. A 82 062306 (2010).

\bibitem{bho}
Brandao, F.G.S.L., and Horodecki, F.
{\it On Hastings' counterexamples to the minimum output entropy additivity conjecture.}
Open Syst. Inf. Dyn. 17, 31 (2010).

\bibitem{cho}
Choi, M.-D.
{\it Completely positive linear maps on complex matrices.}
Linear Alg. Appl. 10, 285 (1975).

\bibitem{cn2}
Collins, B. and Nechita, I.
{\it Random quantum channels II: Entanglement of random subspaces, R\'enyi entropy estimates and additivity problems.} 
Advances in Mathematics 226, 1181-1201(2011).

\bibitem{fal}
Fujiwara, A. and Algoet, P.
{\it One-to-one parametrization of quantum channels.}
Phys. Rev. A 59 3290 (1999).

\bibitem{fuk}
Fukuda, M.
{\it Revisiting additivity violation of quantum channels},
Comm. Math. Phys, DOI 10.1007/s00220-014-2101-2.

\bibitem{fk}
Fukuda, M. and King, C.
{\it Entanglement of random subspaces via the Hastings bound.}
J. Math. Phys. 51, 042201 (2010).

\bibitem{Gruber} Gruber, P.M. {\it Aspects of approximation of convex bodies}, Handbook of Convex Geometry A (P.M. Gruber and J.Wills, eds.), North-Holland, Amsterdam (1993).

\bibitem{has09}
Hastings, M. B.
{\it Superadditivity of communication capacity using entangled inputs.}
Nature Physics 5, 255 (2009).

\bibitem{hwi}
Hayden, P. and Winter, A. 
{\it Counterexamples to the maximal p-norm multiplicativity conjecture for all $p>1$.} 
Comm. Math. Phys. 284 (2008), no. 1, 263--280.

\bibitem{jam}
Jamio{\l}kowski, A.
{\it Linear transformations which preserve trace and positive semi-definiteness of operators.}
Rep. Math. Phys. 3, 275 (1972).

\bibitem{KR}
King, C. and Ruskai, M.B. {\it Minimal Entropy of States Emerging from Noisy Quantum Channels} IEEE Trans. Info. Theory 47, 192–209 (2001).
\bibitem{lin}
Lindblad, G.
{\it A General No-Cloning Theorem.}
Lett. in Math. Phys. 47 (2), 189--196 (1999).

\bibitem{rus}
Ruskai, M.-B.
{\it Qubit Entanglement Breaking Channels.}
Rev. Math. Phys. 15, 643-662 (2003).

\bibitem{sho02}
Shor, P. W. 
{\it Additivity of the classical capacity of entanglement-breaking quantum channels.} 
J. Math. Phys. 43, 4334 (2002).

\bibitem{wol}
Wolf, M. 
{\it Quantum channels \& operations: Guided tour.}
Lecture notes available  \href{http://www-m5.ma.tum.de/foswiki/pub/M5/Allgemeines/MichaelWolf/QChannelLecture.pdf}{online},
July 2012.

\end{thebibliography}
\end{document}